\documentclass[runningheads]{llncs}

\usepackage{amssymb}
\usepackage{amsmath}
\usepackage{mathtools}
\usepackage{relsize}
\usepackage{comment}
\usepackage{color}
\usepackage{cite}
\usepackage{comment}
\usepackage{tcolorbox}
\usepackage{xspace}

\usepackage{dsfont}
\usepackage{graphicx}
\usepackage{nicefrac}
\usepackage{tikz}
\usetikzlibrary{calc,math,patterns,shapes,backgrounds}
\usepackage{todonotes} 
\setuptodonotes{inline}
\usepackage{tabularx}

\newtheorem{Reduction}{Reduction Rule}

\usepackage{framed}
\usepackage{url}
\usepackage{cite}

\usepackage{xcolor}
\usepackage{comment}

\newcommand{\odc}{\textsc{Odd Coloring}}  

\newcommand{\oc}{\textsc {Odd Coloring}}

\usepackage{enumitem}
\usepackage{xcolor}
\usepackage{comment}
\usepackage{tikz}
\usepackage{epsfig}
\usepackage{wrapfig}
\usetikzlibrary{shapes,backgrounds, patterns}

\usepackage{xcolor}
\usepackage{comment}
\usepackage{tcolorbox}
\tcbuselibrary{breakable}

\colorlet{bscolor}{blue}
\colorlet{vrcolor}{red}

\newcommand{\Omit}[1]{}

\usepackage{hyperref,lipsum}

\hypersetup{
  colorlinks   = true, 
  urlcolor     = black, 
  linkcolor    = blue, 
  citecolor   = red 
}
\usepackage{array}

\usepackage{xcolor}
\usepackage{comment}
\usepackage{tcolorbox}
\tcbuselibrary{breakable}
\newtcolorbox{mybox}[2][]{colbacktitle=white,colback=white,coltitle=black,title={#2},fonttitle=\bfseries,#1, left = 1mm, right = 2mm, breakable}
\usepackage{todonotes}
\usepackage[linesnumbered,ruled,vlined]{algorithm2e}
\usepackage{algpseudocode}

\newtheorem{reduction rule}{Reduction Rule}

\newtheorem{observation}{Observation}

\newcommand{\NP}{\ensuremath{\sf{NP}}\xspace}

\newcommand{\OO}{{\mathcal O}}

\newcommand{\fpt}{\ensuremath{\sf{FPT}}\xspace}

\newtheorem{claim claim}{Claim}

\title{On the Parameterized Complexity of Odd Coloring}
\titlerunning{}

\author{Sriram Bhyravarapu\inst{1}, Swati Kumari\inst{2} \and
I. Vinod Reddy\inst{2} 
}
\authorrunning{S. Bhyravarapu and I. Vinod Reddy}
\institute{The Institute of Mathematical Sciences, HBNI, Chennai, India \and
Department of Computer Science and Engineering,  IIT Bhilai, India
\\
\email{sriramb@imsc.res.in, swatik@iitbhilai.ac.in, vinod@iitbhilai.ac.in}}

\begin{document}

\maketitle

\begin{abstract}

A proper vertex coloring of a connected graph $G$ is called an odd coloring if, for every  vertex $v$ in $G$, there exists a color that appears odd number of times in the open neighborhood of $v$. The minimum number of colors
required to obtain an odd coloring of $G$ is called the \emph{odd chromatic
number} of $G$, denoted by $\chi_{o}(G)$.
Determining  $\chi_o(G)$ known to be ${\sf NP}$-hard. Given a graph $G$ and an integer $k$, the \odc{} problem is to decide whether $\chi_o(G)$ is at most $k$. In this paper, we study the parameterized complexity of the problem, particularly with respect to structural graph parameters. We obtain the following results:
\begin{itemize}
    \item We prove that the problem admits a polynomial kernel when parameterized by the  distance to clique.
    \item We show that the problem cannot have a polynomial kernel when parameterized by the vertex cover number  unless ${\sf NP} \subseteq {\sf Co {\text -} NP/poly}$.
    \item  We show that the problem is fixed-parameter tractable when parameterized by distance to cluster, distance to co-cluster, or neighborhood diversity.
    \item We show that the problem is ${\sf W[1]}$-hard parameterized by clique-width.
\end{itemize}
   Finally, we study the complexity of the problem on restricted graph classes. We show that it can be solved in polynomial time on cographs and split graphs but remains NP-complete  on certain subclasses of bipartite graphs.

\end{abstract}
\section{Introduction}
All graphs considered in this paper are finite, connected, undirected, and simple. For a graph $G=(V,E)$, the vertex set and edge set of $G$ are denoted by $V(G)$ and $E(G)$  respectively. 
For a positive integer $k$, a $k$-coloring of a graph $G$ is a mapping $f: V(G) \rightarrow \{1,2,\ldots, k\}$.
A $k$-coloring of a graph $G$ is called a \emph{proper $k$-coloring}, if $f(u) \neq f(v)$ for every edge $uv$ of $G$. The chromatic number $\chi(G)$ of a graph $G$ is the minimum $k$ for which there exists a proper $k$-coloring of $G$.
A hypergraph $\mathcal{H}=(\mathcal{V},\mathcal{E})$ is a generalization of a graph, where hyper-edges are subsets of $\mathcal{V}$ of arbitrary positive size. 
Several notions of vertex coloring have been studied on hypergraphs~\cite{pach2009conflict,nagy2008online}. One such coloring is conflict-free (CF) coloring of hypergraphs, introduced by Even et al.~\cite{even2003conflict} in geometric setting, motivated by its applications in frequency assignment problems on cellular networks.
A coloring of a hypergraph $\mathcal{H}$ is called \emph{conflict-free} if, for every hyper-edge $e$, there is a color that occurs exactly once on the vertices of $e$. 
The minimum number of colors needed to
conﬂict-free color the vertices of a hypergraph $\mathcal{H}$ is called the conﬂict-free chromatic number of $\mathcal{H}$. Conﬂict-free coloring of  hypergraphs is well studied on rectangles~\cite{ajwani2007conflict}, intervals~\cite{bar2008deterministic},  unit disks~\cite{lev2009conflict} etc. 

Cheilaris~\cite{Cheilaris} studied the CF coloring of graphs with respect to neighborhoods. A coloring $f$ (not necessarily proper)  of a graph $G$ is called \emph{conflict-free} (CF) if for every vertex $v$ there exists a color that appears exactly once in the open neighborhood of $v$. For more details on Conflict-free coloring of graphs, see, e.g.,~\cite{abel2018conflict,gargano2015complexity}.

Another variant of conflict-free coloring called \emph{proper conflict-free}  (PCF) coloring was introduced by Fabrici et al.~\cite{fabrici2023proper}. A proper $k$-coloring of a graph $G$ is called a proper conflict-free $k$-coloring if each vertex has a color appearing exactly once in its open neighborhood. 
It was shown~\cite{fabrici2023proper} that 
a planar graph has a proper conflict-free $8$-coloring and also constructed a planar graph with no proper conflict-free $5$-coloring. This problem has been studied on several graph classes, such as minor-closed families,
graph with bounded layered treewidth~\cite{liu2022proper}, and graphs with bounded expansion~\cite{hickingbotham2022odd}.



The concept of odd coloring was introduced by Petru{\v{s}}evski and {\v{S}}krekovski \cite{petruvsevski2022colorings}, which is a generalization of proper conflict-free coloring.  
A proper $k$-coloring of a graph $G$ is called an \emph{odd $k$-coloring} if, every  vertex has some color that appears an odd number of times in its open neighborhood.
The \emph{odd chromatic number} of a graph $G$, denoted $\chi_o(G)$, is the smallest integer $k$ such that $G$ has an odd $k$-coloring. Given a graph $G$ and an integer $k$, the \odc{} problem is to decide whether $\chi_o(G)$ is at most $k$. 



The \odc{} problem has been extensively studied on planar graphs~\cite{caro2022remarks,petr2023odd,petruvsevski2022colorings,cranston2023note,cho2023odd}. 
Petru{\v{s}}evski and {\v{S}}krekovski \cite{petruvsevski2022colorings} showed that $\chi_o(G) \leq 9$ for every connected planar graph $G$ and  conjectured that for all planar graphs $G$, $\chi_o(G) \leq 5$. Towards resolving this conjecture, Petr and Portier~\cite{petr2023odd} showed that $\chi_o(G) \leq 8$ for every planar graph $G$.
Cranston et al.~\cite{cranston2023note} proved that every $1$-planar graph $G$\footnote{it can drawn in the plane so that each edge is crossed by at most one other edge}, $\chi_o(G) \leq 23$. 
Caro et al. \cite{caro2022remarks} showed that for every outerplanar graph $G$, $\chi_o(G) \leq 5$. They also showed that determining $\chi_o(G)$ of a graph $G$ is $\NP$-hard. Furthermore, Ahn et al.~\cite{ahn2022proper} proved that it is $\NP$-complete to decide whether $\chi_o(G) \leq k$ for $k \geq 3$, even when $G$ is bipartite. 

	\vspace{-0.1cm} 
\begin{framed}
	\vspace{-0.2cm}
	\noindent \textsc{\odc{}} \\
	\textbf{Input:} A graph $G=(V,E)$ and integer $k$. \\
	\textbf{Question:}  Does there exists an odd coloring $f: V(G) \rightarrow \{1,\ldots, k\}$ of $G$? 
	\vspace{-0.2cm}
\end{framed}
	\vspace{-0.3cm}

In this paper, we study the parameterized complexity of \odc{}. In parameterized complexity, the running time of an algorithm is measured as a function of input size and a secondary measure called a parameter. A parameterized problem is called  fixed-parameter tractable ($\fpt$) with respect to a parameter $d$, if the problem can be solved in $f(d) n^{O(1)}$ time, where $f$ is a computable function independent of the input size $n$ and $d$ is a parameter associated with the input.  
For more details on parameterized complexity, we refer the reader to the texts~\cite{cygan2015parameterized,downey1999parameterized}. The most natural parameter for \odc{} is $k$, the number of colors.   Since the problem is $\NP$-hard even when $k=2$, it becomes Para-$\NP$-hard when parameterized by number of colors $k$. Therefore, we study the problem with respect to structural graph parameters.

The graph parameter tree-width is a well-known structural graph parameter. 
\odc{} 
is 
 expressible in counting monadic second-order logic (MSO). 
Therefore, the problem is $\fpt$ from the Courcelle~\cite{courcelle1992monadic} meta theorem.
Next, we consider the parameter clique-width~\cite{courcellecw}, which is a generalization of tree-width. We show that the problem is {\sf W[1]}-hard parameterized by the clique-width.

Next, we consider the parameters distance to cluster, distance to co-cluster and neighborhood diversity. These are intermediate parameters between vertex cover and clique-width.
We show that the problem is $\fpt$ when parameterized by (a) distance to cluster, (b) distance to co-cluster, or (c) neighborhood diversity. 
We also study the kernelization complexity of the problem. We show that \odc{} admits a polynomial kernel when parameterized by distance to clique. However, it does not admit a polynomial kernel when parameterized by the vertex cover number  unless ${\sf NP} \subseteq {\sf Co {\text -} NP/poly}$. 

Finally, we consider the parameters distance to split graphs and distance to cographs, both of which generalize  vertex cover and distance to clique. However, the complexity of \odc{} is not known on split graphs and cographs. 
Hence, we study the problem on  cographs and split graphs. We show that the problem can be solved in polynomial time on  cographs and split graphs.  Since \odc{} is NP-complete on bipartite graphs~\cite{ahn2022proper}, we study its complexity on subclasses of bipartite graphs. We show that the problem remains hard even for these subclasses. 


\medskip
\noindent 
\textbf{Organization of the paper.} 
In Section~\ref{gc-dc}, we present results on the parameterized complexity of \odc{}. In Section~\ref{sec:restricted}, we present the results of \odc{} on restricted graph classes. 

\section{Preliminaries}
For $k \in \mathbb{N}$, we use $[k]$ to denote the set $\{1,2,\ldots,k\}$.  
We use $n$ to denote the number of vertices and $m$ to denote the number of edges of the graph.  
For simplicity, the edge between two vertices $x$ and $y$ is denoted by $xy$.  
For a subset $X \subseteq V(G)$, the graph $G[X]$ denotes the subgraph of $G$ induced by the vertices in $X$.  Similarly, $G - X$ represents the graph obtained by removing all vertices in $X$ and their incident edges from $G$. For a  subset $X \subseteq V(G)$,  $E[X]$ denotes the set of edges having both end vertices in the set $X$ in $G$. For subsets $X, Y \subseteq V(G)$, $E[X,Y]$ denotes the set of edges connecting $X$ and $Y$ in $G$. 
For a function $f:X\rightarrow Y$, we denote $f(X')=\bigcup_{x\in X'}f(x)$ where $X'\subseteq X$.

Let $f:V(G) \rightarrow [k]$ be a coloring of a graph $G$. We say that, a vertex $v\in V(G)$ has a color $c\in [k]$ as its \emph{odd color} with respect to $f$, if the number of neighbors of $v$  that are assigned the color $c$ is odd. Formally, $c$ is an odd color of $v$ if $|f^{-1}(c) \cap N(v)|$ is odd.  


\emph{Parameterized Complexity and Kernelization:}
A parameterized problem $L$ is a set of instances $(x,d) \in \Sigma^* \times \mathbb{N}$,  where 
 $\Sigma$ is a finite alphabet and 
 $ d \in \mathbb{N}$ is a parameter. 
We say that 
a parameterized problem $L$ is {\it fixed-parameter tractable} ($\fpt$) if there exists an algorithm  and a computable function $f : \mathbb{N} \rightarrow \mathbb{N}$ such that given $(x,d) \in \Sigma^*\times \mathbb{N}$ the algorithm correctly decides  whether $(x,d) \in L$ in
$f(d)· |(x,d)|^{O(1)}$ time.

A {\it kernelization} algorithm is a polynomial time algorithm that takes an instance $(x, d)$ of a parameterized
problem $L$ as input and outputs an equivalent instance $(x', d')$ of $L$ such that $d'\leq h(d)$ for some computable function $h$ and $d'\leq d$. If
$h$ is a polynomial, we say that $(x', d')$ is a polynomial kernel.

\section{Results: Parameterized Complexity}\label{gc-dc}
\subsection{Parameterization by Distance to Clique}
For a graph $G$, the distance to clique of $G$ is the smallest number $d$ such that there
is a set $X \subseteq V(G)$ of size $d$ with $G- X$ being a clique. 
In this section, we present a polynomial kernel for \oc{}~ parameterized by the distance to clique. 
\begin{theorem}\label{thm:c-kernel}
\oc{} parameterized by the distance to clique $d$ admits a kernel with $O(d^3)$ vertices. 
\end{theorem}

\begin{proof}
	Let $(G,X, d, k)$ be an instance of \oc{}, where
	$X \subseteq V(G)$ is a set of size $d$ 
 such that the graph induced on $C=V(G) \setminus X$ is a clique of size $n-d$. 
	We assume that $k \geq |C|$, otherwise the given instance is a trivial no-instance as we need $|C|$ colors to color the clique $G[C]$.  Further, we assume 
 that $|C| > d^2+d+1$, otherwise we trivially get a kernel of size at most $d^2+2d+1$. 
 
 Let $C_N=\{v \in C~|~ N(v) \cap X  =\emptyset\}$. If $|C_N| \geq d$, 
 then the given instance is a yes-instance and we obtain a trivial kernel. 
 Therefore, we assume that $|C_N|<d$. 
	We  partition the vertices of $X$ based on their degree in the clique $G[C]$. Let $X_{\ell}= \{x \in X : |N(x) \cap C| \leq d-1\}$, $X_{m}= \{x \in X : d \leq |N(x) \cap C| \leq n-d^2-d-1\}$ and $X_h= \{x \in X : |N(x) \cap C| \geq n-d^2-d\}$. 
	Let $X^m_{\ell}=\{x \in X_{\ell}~|~ N(x) \cap X_m \neq \emptyset\}$.  
	\begin{Reduction}
		Delete the vertices of $X_{m}$ from $G$ and for each vertex $x \in X^m_{\ell}$ add a pendant vertex $x'$ adjacent to $x$.  The new instance is $(G + X'_{\ell}-{X_{m}}, (X \cup {X'_{\ell}}) \setminus X_m, d-|X_m|+|X'_{\ell}|, k)$, where $X'_{\ell}$ denotes the set of pendant vertices. 
  
	\end{Reduction}
	\begin{lemma}
		Reduction Rule~1 is safe.
	\end{lemma}
	\begin{proof}
		\emph{Forward direction.} Let $f$ be a $k$-odd coloring of $G$. Let $G'= G+ X'_{\ell}-{X_{m}}$. We define a $k$-coloring $g$ of $G'$  as follows: 
  \begin{itemize}
      \item for each  $v \in X_{\ell}\cup X_h \cup C$, set $g(v)=f(v)$, and  
      \item for each $v \in X'_{\ell}$,  assign $g(v)=c$,   
      where $c$ is an arbitrary color in $[k]-f(N[u])$, and $u$ is the unique neighbor of $v$ in $G'$.  
     
  \end{itemize}
 
		Clearly, $g$ is  a proper $k$-coloring of $G'$.
		Next we argue that $g$ is an odd coloring of $G'$. Consider a vertex $v \in X_h$.  Since $v$ has at least $d$ neighbors in $C$, there exists at least one color in $g(N(v))$ that appears exactly once in the neighborhood of $v$. Thus $v$ has an odd color with respect to $g$. 
  Consider a vertex $v \in C$. As $|C|>d^2+1$, $v$ also has an 
  odd color with respect to $g$. Consider a vertex $v \in X_{\ell}$. 
		If $v \in X^m_{\ell}$, then the color of the pendant vertex adjacent to $v$ appears exactly once in the neighborhood of $v$ and $v$ has an odd color with respect to $g$.   
  If $v \notin X^m_{\ell}$, then, we have that  $f(N(v))=g(N(v))$, and $v$ has an odd color with respect to $g$. 
  Finally, if $v \in X'_{\ell}$, then $v$ has degree one and therefore trivially has an odd color. From the above cases, we conclude that $g$ is an odd coloring of $G'$.

		\emph{Reverse direction.}
		Let $g$ be a $k$-odd coloring of $G'$.
		We define a $k$-coloring $f$ of $G$ as follows: 
  \begin{itemize}
      \item for each $v \in X_{\ell} \cup X_h \cup C$, we assign $f(v)=g(v)$,  
     
     \item  for $v \in X_m$, we have  
      $|N(v) \cap C| \leq n-d^2-d-1$. That is, each vertex $v \in X_m$ is not adjacent to at least $d^2+1$ vertices in $C$.   As $|X_{\ell}|+|X_h|+|N(X_{\ell})\cap C| \leq (d-|X_m|)+d(d-1)$, 
      we color the vertices of $X_m$ with $|X_m|$ many distinct colors, which are different from the colors used in $X_\ell \cup X_h \cup (N(X_{\ell}) \cap C) \cup (N(X_{m}) \cap C)$  in $f$. 
      
  \end{itemize}
  Next we argue that $f$ is an odd coloring of $G$.
		For each $v \in X_m$, we have 
		$|N(v) \cap C| \geq d$, therefore, there is at least one color used for vertices of $C$ appears exactly once in $N(v)$. 
  Hence $v$ has an odd color. 
		For each $v \in X^m_{\ell}$, every color in $f(N(v) \cap X_m)$ acts as a odd color. For each $v \in X_{\ell} \setminus X^m_{\ell}$, $f(N(v))=g(N(v))$, so $v$ retains its odd color. 
  \qed
	\end{proof}
	
	We assume that $X_h\neq \emptyset$; otherwise, by applying the above reduction rule and using the properties of $X_{\ell}$, we trivially get a kernel of size at most $d(d-1)+d$.  
	
	Let $D_h=\bigcap\limits_{x \in X_h}(N(x) \cap C)$ and $D_{\ell}=\bigcup\limits_{x \in X_{\ell}}(N(x) \cap C)$. It is easy to see that $|D_{\ell}| \leq d(d-1)$. Let $C_1=D_h \setminus D_{\ell}$.

	\begin{lemma}~\label{lem-size}
		The size of $D_h$ is at least $n-(|X_h|d+1)d$.
	\end{lemma}
	\begin{proof}
		As each vertex of $X_h$ has at least $n-d^2-d$ neighbors in $C$, the
		number of edges $|E[X_h,C]|$ between $X_h$ and $C$ is at least $|X_h|(n-d^2-d)$. 
		Suppose $|D_h|= n-(|X_h|d+1)d-r$ for some $r \geq 1$. 
  Then the number of edges $|E[X_h,D_h]|$ between $X_h$ and $D_h$ equal to  $|X_h|(n-(|X_h|d+1)d-r)$. Since each vertex in $C \setminus D_h$  has at most $|X_h|-1$ neighbors in $X_h$,  the number of edges $|E[X_h,C\setminus D_h]|$ between $X_h$ and $C \setminus D_h$ is at most $(|X_h|-1)(|X_h|d^2+r)$. Hence, we have
%
	\begin{equation*}
	\begin{aligned}
		|E[X_h,C]| \leq {} & |E[X_h,D_h]|+ |E[X_h,C\setminus D_h]| \\
	& \leq  |X_h|(n-(|X_h|d+1)d-r)+(|X_h|-1)(|X_h|d^2+r) \\
	& =|X_h|(n-d^2-d)-r
	\end{aligned}
	\end{equation*}
	which is contradiction to the fact that $|E[X_h,C]| \geq |X_h|(n-d^2-d)$. 
 \qed
	\end{proof}
	
    Therefore, we have $|D_h| \geq n-(|X_h|d+1)d$. 
  
    Hence $|C_1| \geq n-(|X_h|d+1)d -d(d-1) \geq n-(|X_h|+1)d^2$. Let $C_2 \subseteq C_1$ be a set of $d+1$ arbitrary vertices. Let $C'=C_1 \setminus C_2$.

	\begin{Reduction}
		 Delete the vertices in $C'$ from $G$ and decrease $k$ by $|C'|$. The new instance is $(G - C', X, d, k-|C'|)$.
	\end{Reduction}
	\begin{lemma}
	Reduction Rule~2 is safe.
\end{lemma}
	\begin{proof}
		 \emph{Reverse direction.}
		 If $g$ is an odd $(k-|C'|)$-coloring of $G - C'$ then we can obtain an odd $k$-coloring $f$ of $G$ from $g$ by coloring $C'$ with $|C'|$ many distinct new colors.

		For the other direction, let $f$ be an odd $k$-coloring of $G$. Assume that $f(X_{\ell} )\cap f(C') = \emptyset$. If not, recolor the vertices in $X_{\ell}$ using  $|X_{\ell}|$ many distinct colors from $f(C_2)$. 
		Observe that $N(C')=X_h \cup (C\setminus C')$, that is, in the coloring $f$ of $G$, each vertex of $C'$ receives a color different from each color in $X_{\ell} \cup X_h\cup (C\setminus C')$.  Therefore $f$ restricted to $G - C'$ is a proper $(k-|C'|)$-coloring. 
  We now argue that $f$ is also an odd coloring of $G - C'$. 	Consider a vertex $v \in V(G) \setminus C'$, if $v \in X_\ell$ then the odd color of $v$ in $G - C'$ is same as the odd color of $v$ in $G$. If $v \in (X_h \cup (C \setminus C'))$, then as $C_2$ has $d+1$ vertices,  at least one color $c$ from $C_2$ is not used in $X_\ell \cup X_h$ by the coloring $f$. This color $c$ acts as an odd color for the vertex $v$. 
  \qed
	\end{proof}
	
	It is easy to see that both reduction rules are applicable in polynomial time. Next, we show that the size of the reduced instance is at most $d^3+2d^2$.
	
	\begin{lemma}
		Let $(G,X,d,k)$ be an instance of \oc{} obtained after applying the above reduction rules. If it is a yes-instance then $|V(G)|\leq d^3+2d^2$. 
	\end{lemma}
	\begin{proof}
			We know that the size of $C'$ is at least $n-(|X_h|+1)d^2-d$. Therefore the number of vertices in $X_h \cup X_{\ell}\cup D_{\ell} \cup (C\setminus C')$ is at most $d+d(d-1)+(|X_h|d^2+d^2)$, which is at most $d^3+2d^2$.  
   \qed
	\end{proof}
\end{proof}
\subsection{Parameterization by Vertex Cover Number}
First, we introduce a notion of polynomial-parameter transformation, which, under 
the assumption that $\NP \nsubseteq co\NP/poly$, allows us to show the non-existence of polynomial kernels.
\begin{definition}\cite{bodlaender2011kernel}
 Let $P$ and $Q$ be parameterized problems. We say that there is a polynomial-parameter transformation from $P$ to $Q$, denoted $P \leq _{ppt} Q$,
   if there exists a polynomial time computable function $f: \{0,1\}^* \times \mathbb{N} \rightarrow \{0,1\}^* \times \mathbb{N}$ and
 a polynomial $p: \mathbb{N} \rightarrow \mathbb{N}$ such that for all $x \in \{0,1\}^* $ and $k \in \mathbb{N}$, the following hold:
 \begin{enumerate}
  \item $(x,k) \in P$ if and only if $(x',k')=f(x,k) \in Q$ 
  \item $k' \leq p(k)$
 \end{enumerate}
 The function $f$ is called a polynomial-parameter transformation from $P$ to $Q$.
\end{definition}
\begin{theorem} \cite{bodlaender2011kernel}
Let $P$ and $Q$ be parameterized problems, and suppose that $P^c$ and $Q^c$ are
the derived classical problems. Suppose that $P^c$
is $\NP$-complete, and $Q^c \in \NP$. Suppose
that $f$ is a polynomial time and parameter transformation from $P$ to $Q$. Then, if $Q$ has a
polynomial kernel, then $P$ has a polynomial kernel.
\end{theorem}
\begin{corollary}
 Let $P$ and $Q$ be parameterized problems whose unparameterized
versions $P^c$ and $Q^c$ are $\NP$-complete. If $P \leq_{ppt} Q$ and $P$ does not have a
polynomial kernel, then $Q$ does not have a polynomial kernel.
\end{corollary}

\begin{theorem}\cite{bodlaender2011cross}
    \textsc{Graph Coloring} parameterized by the size of a vertex cover does not
admit a polynomial kernel  unless 
${\sf NP} \subseteq {\sf Co {\text -} NP/poly}$. 
\end{theorem}

\begin{theorem}
 \oc{} parameterized by vertex cover number does not
admit a polynomial kernel, unless 
${\sf NP} \subseteq {\sf Co {\text -} NP/poly}$. 
\end{theorem}
\begin{proof}
 We prove the result by showing a polynomial-parameter transformation from the 
  \textsc{Graph Coloring}  problem, 
parameterized by vertex cover number to \oc{} parameterized by vertex cover number for any constant $k  > 2$. Let $(G,X_g,k,d)$ be an instance of \textsc{Graph Coloring}, where $X_g$ is a vertex cover of $G$ of size $d$, and $k$ is the number of allowed colors.
 We construct an instance $(H,X_h,k',d')$ of \oc{} as follows. Without loss of generality, we assume that $|V(G)|$ is odd, otherwise, we will add an isolated triangle to $G$, which may increase the size of the vertex cover by two.  
\begin{enumerate}
 \item Initialize $H$ as a copy of $G$.

 \item Let $I_o$ be the set of vertices in $V(G) \setminus X_g$ that have odd degree in $G$.  Without loss of generality, assume that $I_o$ is odd; if not, we will add an isolated edge to $G$, which may increase the size of the vertex cover by one. Add a new vertex $z$ and make it adjacent to every vertex of $I_o$.
  \item For each odd degree vertex $v \in X_g$, attach a pendant vertex $v'$. 
 \item Finally, add a vertex $u$ and make it adjacent to every vertex in $V(G)$. 
\end{enumerate}
 It is easy to verify that all vertices of $H$ are of odd degree. The vertex cover of $H$ is $X_g \cup\{u,z\}$ which has size at most $d+2$.
 We show that $(G,X_g,k,d)$ is yes instance of \textsc{Graph Coloring} if and only of $(H,X_h,k',d')$ is a yes instance of \textsc{Odd Coloring}, where $k'=k+1$. 

 \emph{Forward direction.} Suppose $G$ has a $k$-coloring $f: V(G) \rightarrow [k]$, we construct a $(k+1)$-odd coloring $g$ of $H$ as follows. For each $v \in V(G)$
 set $g(v)=f(v)$. For all other vertices  set $g(v)=k+1$. It is easy  to see that $g$ is a proper coloring of $H$. Since every vertex of $H$ has odd degree, $g$ is also an odd coloring of $H$.

 \emph{Reverse direction.} Suppose $H$ has a $(k+1)$-coloring $g: V(H) \rightarrow [k+1]$. Without loss of generality, assume  $g(u)=k+1$. Then, for each vertex $v \in V(G)$, $g(u) \neq g(v)$. Therefore, $g$ restricted to the vertices of $V(G)$ gives a $k$-coloring of $G$.  
\qed
\end{proof}

\subsection{Parameterization by Clique-width}\label{sec:cw}
The \emph{clique-width} of a graph $G$  denoted by $cw(G)$, is defined as the minimum number of labels required to construct $G$ 
using the following four operations:
\begin{enumerate}
	\setlength{\itemsep}{1pt}
	\setlength{\parskip}{0pt}
	\item [(i).] \emph{Introducing a vertex.} $\Phi=v(i)$, creates a new vertex $v$ with label $i$. $G_{\Phi}$ is a graph consisting a single vertex $v$ with label $i$. 
	\item [(ii).] \emph{Disjoint union.} $\Phi=\Phi' \oplus \Phi''$,  $G_{\Phi}$ is a disjoint union of labeled graphs $G_{\Phi'}$ and $G_{\Phi''}$
	\item [(iii).]\emph{Introducing edges.} $\Phi=\eta_{i,j}(\Phi')$, connects every vertex with label $i$ to every vertex with label $j$ ($i \neq j$) in $G_{\Phi'}$.
	\item [(iv).]  \emph{Renaming labels.} {\bf $\Phi= \rho_{i\rightarrow j}(\Phi')$}: each vertex of label $i$ is changed to label $j$ in $G_{\Phi'}$.
\end{enumerate}
An  expression build using these  four operations using $w$ labels is called as $w$-$expression$. 
In other words, the {\it clique-width} of a graph $G$, is the minimum $w$ for which there exists a $w$-expression that defines the graph $G$. 
A $w$-expression $\Psi$ is a \emph{nice} $w$-expression of $G$, if no edge is introduced more than once in $\Psi$.

\begin{theorem}[\cite{fomin2010intractability}]
    The \textsc{Graph Coloring} problem is $W[1]$-hard when parameterized by clique-width.
\end{theorem}

\begin{theorem}
 \oc{} is $W[1]$-hard when parameterized by clique-width.
\end{theorem}
\begin{proof}
 We give a parameterized reduction from \textsc{Graph Coloring} parameterized by clique-width. 
 Given an instance $(G,w,k)$ of  \textsc{Graph Coloring}, where $w$ denotes the clique-width of $G$, we  construct a graph $H$ from $G$ by adding a pendant vertex to every even-degree vertex of $G$. As a result, every  vertex in $H$ has odd degree.

 Let $\Phi_G$ be a $w$-expression of $G$. We obtain a $(w+1)$-expression $\Phi_H$  of $H$ from $\Phi_G$  as follows: for each even-degree vertex 
$v$ in $G$, replace the  sub-expression $i(v)$ in $\Phi_G$ with $\eta_{w+1,i}(w+1(v') \oplus i(v))$, where $v'$ is the pendant vertex adjacent to $v$ in $H$. Thus, clique-width of $H$ is at most $w+1$. Since every vertex in $H$ has odd degree, we have $\chi(G)=\chi_o(H)$. 
 That is, $(G,w,k)$ is yes instance of \textsc{Graph Coloring} if and only if  $(H,w+1,k)$ is a yes instance of \oc{}. 
 \qed
 \end{proof}

\subsection{Parameterization by Neighborhood diversity}
In this section we present a fixed-parameter tractable algorithm for \odc{} parameterized
by the  neighborhood diversity.
\begin{definition}[Neighborhood Diversity~\cite{lampis2012algorithmic}]\label{def:nd}
Let $G=(V,E)$ be a graph. 
Two vertices $u,v\in V(G)$ are said to have the \emph{same type} if and only if $N(u)\setminus \{v\}=N(v)\setminus \{u\}$. A graph $G$ has neighborhood diversity at most $t$ 
if there exists a  partition of $V(G)$ into at most $t$ sets  
$V_1, V_2, \dots, V_t$ 
such that all vertices in each set have same type. 
\end{definition}
Observe that each $V_i$   either forms a  clique or an independent set in $G$.
Also, for any two distinct types
$V_i$ and $ V_j$, 
either each vertex in $V_i$ is adjacent to each vertex in $V_j$,
or no vertex in $V_i$ is adjacent to any vertex in $V_j$. 
We call a set $V_i$ as a \emph{clique type} (resp, independent type) if 
$G[V_i]$ is a clique (resp, independent set). If a type $V_i$ contains exactly one vertex then we consider it as an independent type. It is known that a smallest partition of $V(G)$ into clique types and independent types can be computed in polynomial time \cite{lampis2012algorithmic}. Hence, we assume that the partition $V_1, V_2, \dots, V_t$ of the graph $G$ is given as input.

\begin{observation}\label{obs:nd1}
    Let $f$ be an odd $k$-coloring of $G$. Let $u$ and $v$ be two vertices in an independent type $V_i$. Then, a color $c$ is an odd color of $u$ with respect to $f$ if and only if $c$ is an odd color of $v$ with respect to $f$. 
\end{observation}
Let $f$ be an odd $k$-coloring of $G$.
For a subset $S \subseteq V(G)$, a  color $c\in [k]$ is called an odd color for  $S$ with respect to $f$ if $c$ is an odd color for every vertex in $S$ with respect to $f$.
\begin{observation}\label{obs:nd-2}
    Let $f$ be an odd $k$-coloring of $G$ and
    let $V_i$ be a clique type. Then, 
    each vertex $u\in V_i$ has its odd color form the set $f(V_i)$.  
\end{observation}

\begin{theorem}\label{thm:nd}
\odc~is FPT when parameterized by the neighborhood diversity. 
\end{theorem}
\begin{proof}
Let $f$ be any optimal odd coloring of $G$. 
From the observation~\ref{obs:nd1}, we know that for each independent type $V_i$, there exists a color $c$ such that  $c$ is an odd color for the set $V_i$ with respect to $f$.

Without loss of generality, we assume that $V_1, \ldots, V_{t_1}$ are independent types and $V_{t_1+1}, \ldots, V_{t_2}$ are clique types of $G$, where $t=t_1+t_2$. Since there are $t_1$ independent types, at most  $t_1$ colors are needed to ensure an odd coloring for every independent type.
Let $c_1,\ldots,c_{t_1}$ be the set of colors (not necessarily distinct) such that $c_i$ is an odd color for the independent type $V_i$ in an optimal odd coloring of $G$. Without loss of generality, assume that
$c_1,\ldots,c_{t_1}$ are distinct. 

For each $i\in [t_1]$, let $T_i\subseteq\{V_1,\ldots, V_{t_1}\}$ be the set of independent types  such that the color $c_i$ is an odd color for each independent type in $T_i$. Let $A_i\subseteq \{V_1, V_2, \dots, V_t\}$  be the set of types (both independent and clique)
such that each type in $A_i$ has a vertex colored $c_i$ in an optimal odd coloring of $G$.  
Further, let $g_i:A_i\rightarrow \{o,e\}$ be a function that assigns a value from $\{o, e\}$ to each type in $A_i$ indicating whether the color $c_i$ is assigned to an odd number of vertices or an even number of vertices in the type. 
That is, 
for each $V_{\ell}\in A_i$,  assign 
$g_i(V_{\ell})=e$, if the number of vertices from $V_{\ell}$  assigned the color $c_i$ is even and otherwise
assign $g_i(V_{\ell})=o$.

Guessing the colors $c_1, \ldots, c_{t_1}$, $T_i$'s, $A_i$'s and $g_i$'s for each $i\in [t_1]$ can be done in FPT time.
The vertices of 
$G$ are colored in three phases as described below.

{\bf Phase~I.} For each $i\in [t_1]$, given a tuple $(c_i, T_i, A_i, g_i)$, 
    we  partially color the vertices of $G$ so that each independent type  in $T_i$ has $c_i$ as an odd color. 
We apply the following coloring scheme to each type $V_j\in A_i$. 
If $g_i(V_j)=e$, 
 arbitrarily choose two uncolored vertices from $V_j$ and assign them the color $c_i$. If $g_i(V_j)=o$,  arbitrarily choose one uncolored vertex from $V_j$ and assign the color $c_i$. If it is not possible to select vertices according to the above process, then declare the tuple $(c_i, T_i, A_i, g_i)$ invalid. 

{\bf Phase~II.} In this phase, we determine whether the colors $c_1, \ldots, c_{t_1}$ can be assigned to other independent type vertices of $G$ by respecting $g_i$'s.
Let $V_j$ be an independent type and  $d_j$ be the number of vertices in $V_j$ that were colored in Phase~I. Depending on the sizes of $V_j$ and $d_j$, we have the following cases.

\begin{itemize}
    
    \item $|V_j|$ is even
    \begin{itemize}
        \item $d_j$ is even and a color $c\in [t_1]$ 
         is assigned to an even number of vertices of $V_j$ during Phase~I: 
       
        Assign the color $c$ to all remaining uncolored vertices in $V_j$. Since $|V_j|-d_j$ is even, the total number of vertices in $V_j$ assigned $c$ remains even.

    \item $d_j$ is even and  no color that is assigned to an even number of vertices in $V_j$ during Phase~I:  
In this case, an even number of colors were used in $V_j$ during Phase~I. 
Arbitrarily pick a color $c$  used in $V_j$ during Phase~I and
assign it to all remaining uncolored vertices in $V_j$. Since $|V_j|-d_j$ is even, the total number of vertices in $V_j$ assigned $c$ reamins odd. 

        \item $d_j$ is odd :
      
        At least one color $c$ was used an odd number of times in $V_j$ during Phase~I.
        Assign $c$ to all uncolored vertices of $V_j$, except for one vertex (to respect $g_j$).
        The one uncolored vertex will be colored in  Phase~III with a new color.
           
    \end{itemize}

    \item $|V_i|$ is odd
    
This case is similar to the previous one; therefore, the details are omitted.

 \end{itemize}

\medskip
{\bf Phase~III.} 
After Phase~I and Phase~II, some vertices of $G$ are colored with the colors from the set $[t_1]$ and some vertices of $G$ remain uncolored.
We construct an auxiliary graph $G'$ from $G$ where $V(G')$  consists of all uncolored vertices of $G$ and $E(G')=\{uv\in E(G)\mid u,v\in V(G')\}$. 
Since $G'$ is a subgraph of $G$, the neighborhood diversity  
of $G'$ is at most $t$.  

\begin{claim}\label{cla:gg'}
    There is an odd coloring of $G$ using $t_1+s$ colors respecting the guess $(c_i, T_i, A_i, g_i)$ for each $i\in [t_1]$, 
    if and only if there is a proper coloring of $G'$ using $s$ colors.  
\end{claim}     

The proof of the above theorem is trivial and hence we omit the proof. 
We use the following theorem from \cite{robert-g-nd} to obtain the chromatic number of $G'$ and thus obtain an odd coloring of $G$. 
\begin{theorem}[\cite{robert-g-nd}]\label{thm:graphcoloring-nd}
   The \textsc{Precoloring Extension} problem can be solved in time
$\OO(q^{2.5q+o(q)} \cdot n)$, where $q = 2^{2t}$, on graphs of neighborhood diversity at most $t$. 
\end{theorem}

This completes the description of the algorithm. We now analyze its running time. 
Guessing the partition of independent types that receive the same odd color, 
$A_i$ and $g_i$ for each $i\in [t_1]$ 
can be done in time $t^t2^{\OO(t^2)}$. The precoloring extension of coloring (if exists) for an independent type can be done 
in polynomial time. The graph $G'$ can be constructed in polynomial time. Obtaining the chromatic number for $G'$ can be done in time $2^{\OO(t2^{2t})}\cdot n^{\OO(1)}$. 
Thus, the overall running time of the algorithm is $2^{\OO(t2^{2t})}n^{\OO(1)}$. 

\end{proof}

\subsection{Parameterization by Cluster Vertex Deletion Set}\label{sec:cvd}
In this section, we show that \odc~is FPT parameterized by the distance to cluster. Given a graph $G$, a set $X\subseteq V(G)$ of size at most $t$ such that $G-X$ is a disjoint union of cliques can be computed in time $1.9102^t n^{\mathcal{O}(1)}$ \cite{d2cfpt}. 
Therefore, we assume that $X$ is given as part of the input. 

\begin{theorem}\label{thm:cvd}
\odc~is FPT when parameterized by the distance to cluster. 
\end{theorem}

\begin{proof}
Let $X \subseteq V(G)$ of size at most $t$   such that $G -X$ is a disjoint union of cliques. We run over all possible colorings  of $X$. For each coloring of $X$, we first check whether it is a proper coloring of $G[X]$.  Discard all colorings which do not form a proper coloring of $G[X]$.  If a coloring is a proper coloring of $G[X]$ then we proceed to check if it can be extended to the entire graph.


Given a proper coloring $c: X \rightarrow [t]$ of $X$, we now check whether there exists an odd coloring $f: V(G) \rightarrow[k]$ such that $f|_X=c$.


  Let $c:X \rightarrow \{1,2, \ldots, t_1 \}$ be a proper  coloring of $G[X]$, where $t_1 \leq t$. For each vertex $v\in X$, we guess its odd color in the coloring $f$. Let $g: X \rightarrow [t']$, be a function where $g(v)$ denotes the odd color of vertex $v\in X$ with respect to $f$. Notice that $t'\leq t_1+ t \leq 2t$, as we may need at most $t$ new colors to act as the odd colors for the vertices of $X$. 


  Let $C_1, \ldots, C_{p}$ be the cliques of the graph $G-X$. For each subset $Y \subseteq X$ and  each clique $C$ in $G-X$ we define: 
  $$T^Y_C=\{v \in C~|~ N(v) \cap X=Y\}. $$

   For a subset $Y \subseteq X$, we define: $T^Y=\bigcup\limits_{C}T^Y_C$. 
   For each $Y \subseteq X$, let $h^Y: [t'] \rightarrow \{e,o,0\}$, where $h^Y(i)$ denotes the 
 parity (even, odd or zero)\footnote{$h^Y(i)=0$ means no vertex of $T^Y$ is colored with $i$ (we call it as zero parity)} of color $i$ in $T^Y$ with respect to the odd coloring $f$. That is $h^Y(i)=e$ (resp. $h^Y(i)=o$ and $h^Y(i)=0$) 
 if $|f^{-1}(i)\cap T^Y|$ is even (resp. odd and zero). 
 Let $h=\{h^Y~|~Y \subseteq X\}$.
    Given the coloring $c$  and the function $h$, 
    it is easy to verify whether each vertex $v\in X$ has an odd color by considering the neighbors of $v$ in $X$ (i.e., $N(v)\cap X\cap c^{-1}(g(v)$) and its neighbors in $G-X$  (i.e., $\sum\limits_{Y:v\in Y} h^Y(g(v))$), the sum of which should be odd. Thus it is sufficient to obtain a proper coloring of $G-X$ that respects $h^Y$ for each $Y\subseteq X$ ensuring that each vertex in $G-X$ has an odd color. 
    
 

 Let $C_1,\ldots, C_{p}$ be an arbitrary and fixed ordering of cliques in $G-X$. 
 For a clique $C_q \in G-X$, and for each subset $Y \subseteq X$, we define the following. 
  
   \begin{itemize}
       \item $h^{Y}_q: [t'] \rightarrow \{0,1\}$, where 
       $h^Y_q(i)=1$ if  there exists a vertex in $T^Y_{C_q}$ that is assigned color $i$ in $f$,  and $h^Y_q(i)=0$ otherwise . 
   \item   $\widehat{h}^Y_q:[t'] \rightarrow \{e,o,0\}$, where $\widehat{h}^Y_q(i)$ denotes the parity of color $i$ in $\bigcup \limits _{j \in [q]} T^Y_{C_j}$ in 
   the coloring $f$. 
   \item  
       $h_q=\{h^{Y}_q~|~Y \subseteq X\}$ and $\widehat{h}_q=\{\widehat{h}^Y_q~|~Y \subseteq X\}$.
       \end{itemize}
    
      For each $q\in [p]$ and for each $Y \subseteq X$, we run over all possible $h_q$ and $\widehat{h}_q$.
    Given a coloring $c$ of $X$,  $h_q$ and $\widehat{h}_q$ for each $q \in [p]$,  we now check whether there exists a odd coloring $f: V(G) \rightarrow[k]$ such that $f|_X=c$ and $f$ respect $h_q$ and $\widehat{h}_q$. For this purpose we design a dynamic programming algorithm.  



\medskip 
\noindent
\emph{Definition of the table entry.} Given a $h_q$ and $\widehat{h}_q$, let $M[q, h_q, \widehat{h}_q]$ denote the minimum number of colors  
 needed in a coloring $\ell_q$ of the graph $G_q=G[C_1\cup, \dots, \cup C_q\cup X]$ such that 
 \begin{itemize}

  \item for each subset $Y \subseteq X$, and for each color $i \in [t']$, $h^Y_q(i)=1$ if and only if $\ell_q^{-1}(i)\cap T^Y_{C_q}\neq \emptyset$.
  \item for each subset $Y \subseteq X$, and for each color $i \in [t']$, 
  $\widehat{h}^Y_q(i)$ is equal to the parity of the set \big($\ell_q^{-1}(i)\cap \bigcup \limits _{j \in [q]} T^Y_{C_j}$\big),  
     \item $\ell_q|_X=c$

  \item each vertex in $C_1\cup, \dots, \cup C_q$ has an odd color with respect to $\ell_q$. 
 \end{itemize}

 We now show how to compute the entry $M[q, h_q, \widehat{h}_q]$. 

 \medskip
 \noindent
\textbf{Base case.}
We set $M[1, h_1, \widehat{h}_1]=a_1$, where 
$a_1=|\ell_1(C_1) \setminus [t']|$ is the number of colors (outside $[t']$) used by $\ell_1$ to color the vertices of the clique $C_1$. Here $\ell_1$ satisfies the above mentioned constrains.  



\medskip
\noindent
\textbf{Recursive step.} For each $q\geq 2$, 
we 
set $$M[q, h_q, \widehat{h}_q]=\max \{a_q, \min\limits_{\substack{h_{q-1},\widehat{h}_{q-1}}} \{M[q-1, h_{q-1}, \widehat{h}_{q-1}]\}\}$$
  where $a_q=|\ell_q(C_q) \setminus [t']| $, 
    and $\widehat{h}_{q-1}$ satisfies the following:         
           
           \begin{equation*}
         \widehat{h}_{q-1}(i)=\begin{cases}
           0 & \text{ if } \widehat{h}^Y_{q}(i)=0\\
           e & \text{ if } \widehat{h}^Y_{q}(i)=e 
           \text{ and } h^Y_q(i)=0, \\
           o  & \text { if } \widehat{h}^Y_{q}(i)=e 
           \text{ and } h^Y_q(i)=1,  \\
            o &  \text { if }  \widehat{h}^Y_{q}(i)=o 
           \text{ and } h^Y_q(i)=0,  \\
            e \text{~or~} 0 &  \text { if } \widehat{h}^Y_{q}(i)=o   
            \text{ and } h^Y_q(i)=1. 
            \\
           \end{cases}
       \end{equation*}
The final coloring is given by the coloring $\ell_p$, which is obtained from the entry $M[p,*,*]$.      

    

\medskip
\noindent
\textbf{Output. } Notice that each tuple $M[p,  h', \widehat{h'}]$ denotes the minimum number of new colors outside $[t']$ that are assigned to the vertices of $G-X$ respecting $h'$ and $\widehat{h'}$, and each vertex in $G-X$ has an odd color. We now check for those entries  $M[p,  *, \widehat{h'}]$ where $\widehat{h'}(i)=h^Y(i)$ for each $i$ and output the corresponding entry value that is the  minimum over all of them. 



We now describe the running time of the algorithm. We go over all functions of $f$, $g$ and $h$ and check if there exists an odd coloring extending it. 
The dynamic programming table stores $p\cdot 3^{2t\cdot 2^t}$ entries. Since each entry can be computed in polynomial time and the overall tome to compute the table is $3^{O(t2^t)}n^{\OO(1)}$. 
\qed
\end{proof}

\subsection{Parameterization by Distance to Co-Cluster}\label{sec:co-cl}

In this section, we show that \odc~ is {\sf FPT} when parameterized by the distance to co-cluster. 

\begin{theorem}\label{thm:dcc}
\odc~is {\sf FPT} when parameterized by the distance to co-cluster. 
\end{theorem}


The algorithm is similar in flavor to the one in Theorem \ref{thm:cvd}. Below, we highlight the key steps.  Let $X \subseteq V(G)$ be a subset of size $t$ such that $G -X$ is the complement of a cluster graph. 
Let $I_1, I_2, \dots, I_p$ be the independent sets of the graph
 $G-X$  such that 
 each vertex in $I_i$ is adjacent to every vertex in $I_j$, for all $i\neq j$. For two subsets  $A, B\subseteq V(G)$, we say that $A$ and $B$ form a \emph{biclique} if both are  independent sets and every vertex in $A$ is adjacent to every vertex of $B$. 
 \begin{itemize}

 \item We run over all possible colorings  of $X$. For each coloring, we first check whether it is a proper coloring of $G[X]$. If not, we discard it. Otherwise, we test whether it can be extended to the entire graph.
 Given a proper coloring $c: X \rightarrow [t_1]$ of $G[X]$, where $t_1\leq t$, 
 it remains to check whether there exists an odd coloring $f: V(G) \rightarrow[k]$ such that $f|_X=c$. 


 \item Let $c:X \rightarrow \{1,2, \ldots, t_1 \}$ and 
     $g: X \rightarrow [t']$ be the colorings, where for each vertex $v\in X$,  
     $c(v)$ denotes the assigned color and $g(v)$ denotes its odd color with respect to $f$. 
     Since $t'\leq t_1+ t \leq 2t$, the number of colors used is at most $2t$ by $g$. 
     
     \item Let $Y=c(X)\cup g(X)$ denote the set of assigned colors and odd colors for the vertices in $X$. 
     We first try to assign colors from $Y$ to the vertices of $G-X$. 
     Since 
     each pair of independent sets in $G-X$ form a biclique, 
     each color in $Y$  can be assigned to vertices of at most one independent set. Thus at most $|Y|$ 
    independent sets contain vertices that receive colors from $Y$. 
     

     \item Let $y\leq |Y|$ be the number of independent sets from $G-X$ that contain vertices assigned colors from $Y$. Let $\mathcal{Z}=\{Z_1, Z_2, \dots, Z_{y}\}$ be the collection of such independent sets. 
     Let $(Y_1, \ldots, Y_p)$ be the partition of colors of $Y$, where $Y_i$ denotes the set of colors used in the independent set $Z_i$. 
     We further refine this partition by 
 distributing the colors in $Y_i$ according to the type partition of $Z_i$.  
 That is, for each color  $d\in Y_i$ assigned within  $Z_i$,  
  we guess the set of types (based on their neighborhood in $X$)
  from $Z_i$, where vertices are assigned the color $d$. We also guess the parity (even or odd) of the number of vertices of each type that are assigned the color $d$. 
  
  \item It is possible to check in polynomial time whether each vertex in $X$ has an odd color satisfying the functions $c$, $g$ and the  partition $\mathcal{Z}$.

    \item Notice that there are $\ell=p-y$ independent sets in which no vertices are assigned colors from $Y$. 
It is easy to verify that at most $\ell+2$ new colors (i.e., colors not in $Y$) are sufficient to color these independent sets while ensuring that every vertex in $G-X$ has an odd color.
The additional two colors may be necessary to satisfy the odd coloring property for vertices in $G-X$. 
We guess whether zero, one or two additional colors are needed along with their parity.  We identify the independent sets that will receive these extra colors and assign them accordingly.




     \item The final step is to identify the independent sets in $\mathcal{Z}$ from $I_1, I_2, \dots, I_p$ using a dynamic programming routine. This step is similar to the algorithm in  Theorem \ref{thm:cvd}. 
 \end{itemize}

\section{Odd Coloring on Restrcited Graph Classes}\label{sec:restricted}
\subsection{Cographs}\label{sec:cog}
In this section, we give a linear-time algorithm for computing odd chromatic number on cographs. A \emph{cograph} is a graph with no induced $P_4$.

\begin{lemma}[\cite{corneil1985linear}]~\label{cograph-structure}
    If $G$ is a cograph then one of the following holds.
    \begin{enumerate}
        \item $G$ has at most one vertex.
        \item $G$ is the disjoint union of two cographs $G_1$ and $G_2$, i.e.,
        $G=G_1 \cup G_2$.
        \item $G$ is the join of two cographs $G_1$ and $G_2$, i.e.,
        $G=G_1 \vee G_2$.
    \end{enumerate}
\end{lemma}
A cotree $T_G$ of a cograph $G$ is a rooted tree in which each internal vertex is either of
$\cup$ (union) type  or $\vee$ (join) type. The leaves of $T_G$ correspond to the vertices of the cograph $G$.

\begin{lemma}
    For any connected cograph $G$, we have $\chi(G) \leq \chi_o(G) \leq \chi(G)+2$.
\end{lemma}
\begin{proof}
    As every odd coloring of $G$ is a proper coloring, we get $\chi(G) \leq \chi_o(G)$.
    Since $G$ is a connected cograph, let $G=G_1 \vee G_2$. Let $f$ be a proper coloring of $G$ with colors $\{1,2,\ldots, \chi(G)\}$. 
    We obtain an odd coloring $g$ of $G$ from $f$ by changing the color of two arbitrary vertices $u$ from $G_1$  and $v$ from $G_2$ with colors $\chi(G)+1$ and $\chi(G)+2$ respectively.
    The new colors $\chi(G)+1$ and $\chi(G)+2$ are odd colors of vertices of $G_2$ and $G_1$ respectively. 
    \qed
\end{proof}
We now introduce the notions of 
strong proper-coloring and strong odd-coloring. 

\begin{definition}[Strong proper-coloring]
A \emph{strong proper-coloring} of a graph $G$ is a proper-coloring $f:V(G) \rightarrow [k]$  such that $|f^{-1}(i)|$ is odd for some $i \in [k]$. The smallest integer  $k$ 
such that $G$ has a strong proper-coloring  is called \emph{strong chromatic number} of $G$, denoted by $\widetilde{\chi}(G)$. 
\end{definition}

\begin{definition}[Strong odd-coloring]
A \emph{strong odd-coloring} of a graph $G$ is an odd-coloring $f:V(G) \rightarrow [k]$  such that $|f^{-1}(i)|$ is odd for some $i \in [k]$. The smallest integer  $k$ 
such that $G$ has a strong odd-coloring  is called \emph{strong odd chromatic number} of $G$, 
denoted by $\widetilde{\chi_o}(G)$. 
\end{definition}

\begin{lemma}\label{cograph-relation}
For any graph $G$, we have 
    \begin{enumerate}
        \item  $\widetilde{\chi}(G) \leq {\chi}(G)+1$
        \ \item $\widetilde{\chi_o}(G) \leq {\chi_o}(G) +1$
    \end{enumerate}
\end{lemma}
\begin{proof}
    Given a proper coloring $f$ of $G$ using $\chi(G)$ colors, 
    we can obtain a strong proper coloring $g$ of $G$ from $f$ using $\chi(G)+1$ colors 
    by changing the color of an arbitrary vertex with the new color $\chi(G)+1$. Therefore we get  $\widetilde{\chi}(G) \leq {\chi}(G)+1$. Analogously, we can prove the  second part of the lemma. 
    \qed
\end{proof}

\begin{lemma}\label{cograph-union}
    Let $G_1$ and $G_2$ be two graphs and  $G = G_1 \cup G_2$. Then $\chi_o(G)=\max\{\chi_o(G_1),\chi_o(G_2) \}$
\end{lemma}
\begin{proof}
   As there are no edges from vertices of $G_1$ to the vertices of $G_2$ in $G$, 
   it follows that  
   $\chi_o(G)=\max\{\chi_o(G_1),\chi_o(G_2) \}$. 
   \qed
\end{proof}
\begin{lemma}\label{cograph-join}
      Let $G_1$ and $G_2$ be two graphs and $G = G_1 \vee G_2$.  Then 

      $\chi_o(G) = \min\{\chi_o(G_1)+\chi_o(G_2),  \widetilde{\chi}(G_1)+\widetilde{\chi}(G_2),\widetilde{\chi_o}(G_1)+{\chi}(G_2),{\chi}(G_1)+\widetilde{\chi_o}(G_2)\}$
\end{lemma}

\begin{proof}
 Let $f_i$ be an odd coloring of $G_i$ with colors $\{1,2, \ldots, \chi_o(G_i)\}$, for each $i \in \{1,2\}$. We define an odd coloring $f: V(G) \rightarrow \{1,2, \ldots, \chi_o(G_1),\chi_o(G_1)+1, \ldots, \chi_o(G_1)+\chi_o(G_2)\}$ of $G$ as follows.
       
 $$
    f(u) = 
        \begin{cases} 
        f_1(u) &\text{if $u \in V(G_1)$;} \\
	\chi_o(G_1)+f_2(u) &\text{if $u \in V(G_2)$  ;}\\	
			\end{cases}
			$$
It is easy to see that $f$ is an odd coloring of $G$ using $\chi_o(G_1)+\chi_o(G_2)$ colors. Therefore, $\chi_o(G) \leq \chi_o(G_1)+\chi_o(G_2)$.

Similarly we can show that (a)
$\chi_o(G) \leq \widetilde{\chi}(G_1)+\widetilde{\chi}(G_2)$, (b) $\chi_o(G) \leq  \widetilde{\chi_o}(G_1)+{\chi}(G_2)$ and (c) $\chi_o(G) \leq  {\chi}(G_1)+\widetilde{\chi_o}(G_2)$. Using the four inequalities,  
we get the following. 
\begin{equation}\label{eq1}
\chi_o(G) \leq \min\{\chi_o(G_1)+\chi_o(G_2),  \widetilde{\chi}(G_1)+\widetilde{\chi}(G_2),\widetilde{\chi_o}(G_1)+{\chi}(G_2),{\chi}(G_1)+\widetilde{\chi_o}(G_2)\}
\end{equation}
\par
Next we show the lower bound.
Let $\chi_o(G)=k$ and
$f$ be an odd coloring of $G$ using $k$ colors. 
For $i \in \{1,2\}$, let $f_i$ be the coloring of $G_i$ obtained by restricting $f$ to the vertices of $G_i$.  Observe that $f_1(V(G_1)) \cap f_2(V(G_2))=\emptyset$ as there are all possible edges between $G_1$ and 
$G_2$ in $G$. Let $|f_1(V(G_1))|=k_1$ and  $|f_2(V(G_2))|=k_2$, where $k_1+k_2=k$. As $f$ is an odd coloring of $G$,  both $f_1$ and $f_2$ are proper colorings of $G_1$ and $G_2$ respectively. Therefore, we have $\chi(G_1) \leq k_1$ and  $\chi(G_2) \leq k_2$. 

We divide the proof into several cases.

 \begin{itemize}
     \item \textbf{Case~1.} $\widetilde{\chi_o}(G_1) \leq k_1$.  
     
     In this case we get 
 $$\chi_o(G)=k=k_1+k_2 \geq \widetilde{\chi_o}(G_1)+\chi(G_2)$$

\item  \textbf{Case~2.} $\widetilde{\chi_o}(G_1) > k_1$, ${\chi_o}(G_1) = k_1$ and $\widetilde{\chi}(G_1) > k_1$.   

Then $f_2$ must be an odd $k_2$-coloring of $G_2$. Otherwise some vertices of $G_2$ in the coloring $f$ will not have any odd color which contradicts  the fact that $f$ is an odd coloring of $G$.
Therefore, in this case we have $\chi_o(G_1) \leq k_1$ and 
$\chi_o(G_2) \leq k_2$. Therefore, we get, $$\chi_o(G)=k=k_1+k_2 \geq \chi_o(G_1)+\chi_o(G_2)$$

\item  \textbf{Case~3.} $\widetilde{\chi_o}(G_1) > k_1$, ${\chi_o}(G_1) > k_1$ and $\widetilde{\chi}(G_1) = k_1$. 
   
Then $f_2$ must be a  strong proper $k_2$-coloring of $G_2$. Otherwise some vertices of $G_1$ in the coloring $f$ will not have any odd color which contradicts   the fact that $f$ is an odd coloring of $G$.
Therefore, in this case we have $\widetilde{\chi}(G_1) \leq k_1$ and 
$\widetilde{\chi}(G_2) \leq k_2$. Therefore we get, $$\chi_o(G)=k=k_1+k_2 \geq \widetilde{\chi}(G_1)+\widetilde{\chi}(G_2)$$

\item  \textbf{Case~4.}  ${\chi_o}(G_1) > k_1$ and $\widetilde{\chi}(G_1) > k_1$.

Then $f_2$ must be a  strong odd $k_2$-coloring of $G_2$. Otherwise some vertices of $G_1$ in the coloring $f$ will not have any odd color which contradicts  the  fact that $f$ is an odd coloring of $G$.
Therefore we get
$$\chi_o(G)=k=k_1+k_2 \geq \chi(G_1)+\widetilde{\chi_o}(G_2)$$
\end{itemize}

Therefore, combining all the above we get 
\begin{equation}\label{eq2}
\chi_o(G) \geq \min\{\chi_o(G_1)+\chi_o(G_2),  \widetilde{\chi}(G_1)+\widetilde{\chi}(G_2),\widetilde{\chi_o}(G_1)+{\chi}(G_2),{\chi}(G_1)+\widetilde{\chi_o}(G_2)\}
\end{equation}
From equations \ref{eq1} and \ref{eq2} the result follows. 
\qed
 \end{proof}

 \begin{lemma}\label{cograph-computation}
       Let $G_1$ and $G_2$ be two graphs and let $G = G_1 \vee G_2$ then 
       \begin{enumerate}
           \item $\chi(G)=\chi(G_1)+\chi(G_2)$
           \item $\chi_o(G) =\min\{\chi_o(G_1)+\chi_o(G_2),\widetilde{\chi_o}(G_1)+{\chi}(G_2),{\chi}(G_1)+\widetilde{\chi_o}(G_2)\}$
           \item $\widetilde{\chi}(G) = \min\{\widetilde{\chi}(G_1)+{\chi}(G_2),{\chi}(G_1)+\widetilde{\chi}(G_2)\}$
           \item $\widetilde{\chi_o}(G) = \min\{\widetilde{\chi_o}(G_1)+{\chi}(G_2),{\chi}(G_1)+\widetilde{\chi_o}(G_2)\}$
       \end{enumerate}
 \end{lemma}       
\begin{proof}
           Follows from the Lemmas~\ref{cograph-relation}  and~\ref{cograph-join}. 
           \qed
\end{proof}

\begin{theorem}
    Given a cograph $G$, we can compute $\chi_o(G)$ in linear time. 
\end{theorem}
\begin{proof}
    Follows from the Lemmas~\ref{cograph-structure},~\ref{cograph-union} and~\ref{cograph-computation}. 
    \qed
\end{proof}
\subsection{Split Graphs}\label{sec:split}
In this section we show that \odc{} can be solved in polynomial time on split graphs. A graph is a split graph if its vertices can
be partitioned into a (maximal) clique and an independent set. We denote a split graph with $G = (K, I)$ where $K$ and $I$ denotes the partition of $V(G)$ into a clique and an independent set.
 Since $K$ is maximal, there is no vertex $v \in I$ such that $N(v)=K$. For a subset $Y \subseteq K$, we use $T^Y=\{v \in I~|~ N(v)=Y\}$ to denote the subset of vertices in $I$, whose neighborhood in $K$ is exactly $Y$.
 Throughout this section, we assume that $K=\{v_1,v_2,\ldots,v_k\}$, $|K|=k$ and $I=\{u_1,u_2,\ldots u_{\ell}\}$, $|I|=\ell$. 

\begin{theorem}
    Let $G=(K,I)$ be a split graph. Then $k \leq \chi_o(G)\leq k+1 $. 
\end{theorem}
\begin{proof}
  As every odd coloring of $G$ is a proper coloring of $G$, we have $\chi_o(G)\geq \chi(G) \geq k$. We can obtain an odd $(k+1)$-coloring of $G$ by coloring the vertices of $K$ with $k$ distinct colors and all vertices $I$ with color $k+1$. Hence,  $\chi_o(G)\leq k+1 $. \qed
\end{proof}

In the remainder of this section, we  give a characterization of split graphs  that have odd chromatic number two. 
\begin{lemma}
     Let $G=(K,I)$ be a split graph with $K=\{v_1,v_2\}$.
     Then, $\chi_o(G)=2$ if and only if $I_i=|N(v_i)\cap I|$ is even for each $i \in \{1,2\}$.
\end{lemma}
\begin{proof}

   \emph{Forward direction.} Let $\chi_o(G)=2$ and let $f$ be an odd $2$-coloring of $G$ with colors $\{1,2\}$. Without loss of generality assume that $f(v_i)=i$ for $i \in \{1,2\}$. Then we must have $f(I_1)=\{2\}$ and $f(I_2)=\{1\}$. As $f$ is an odd $2$-coloring of $G$, 
   both $|I_1|$ and $|I_2|$ must be even.
    
    \emph{Reverse direction.} As there is no vertex $u_j \in I$ such that $N(u_j)=K$ we have $I_1 \cap I_2=\emptyset$. If both $|I_1|$ and $|I_2|$ are even then odd $2$-coloring of $G$ is obtained by coloring $v_1$ and every vertex of $I_2$ with color $1$ and $v_2$ and every vertex of $I_1$ with color $2$. Hence, $\chi_o(G)=2$. \qed
\end{proof}
For the remainder of this section, we assume that $|K| \geq 3$. Notice that in any proper coloring $f$ of a split graph $G=(K,I)$, every vertex of $I$ has an odd color with respect to $f$. Therefore, in all proofs presented hereafter, we omit the proof of showing a vertex in $I$ has an odd color with respect to a proper coloring.  
\begin{lemma}
      Let $G=(K,I)$ be a split graph. If there exists a  vertex $v_p$ in $K$ such that $N(v_p) \cap I =\emptyset$ then $\chi_o(G)=k$.
\end{lemma}
\begin{proof}
    Let $f:V(G) \rightarrow[k]$ be a proper coloring of $G$ defined as follows. For each $v_i \in K$, set $f(v_i)=i$  and for every $u_j \in I$, set $f(u_j)=p$. As $|K|\geq 3$, every vertex $v_i$ of $K$ have color $q$ as an odd color with respect to $f$, where $q\in [k]\setminus \{i,p\}$. Thus, $\chi_{o}(G) \leq k$. \qed
\end{proof}
For the remainder of this section,  $N(v_i) \cap I \neq \emptyset$ for all $v_i \in K$.
\begin{theorem}
 Let $G=(K,I)$ be a split graph. Then, 
 $\chi_o(G)=|K|+1$ if and only if there exists a vertex $v \in K$ such that for every $w \in K -\{v\}$, 
 $|T^{K-\{w\}}|$ is odd and $N(v) \cap I = \bigcup\limits_{w\in K\setminus \{v\}}  T^{K-\{w\}}$. 
\end{theorem}
\begin{proof}
 \emph{Reverse direction.} Let $|K|=k$. Suppose there exists a vertex $v \in K$ satisfying the conditions of the premise. We will show that $\chi_o(G)=k+1$.
 
Assume for contradiction that $\chi_o(G)=k$. Let $f:V(G) \rightarrow[k]$ be an odd $k$-coloring of $G$. 
 Since $f$ is a proper coloring, no two vertices of $K$ are assigned the same color, so $f(v_i)=i$ for $v_i \in K$. Now, for every vertex $w \in K -\{v\}$, all the vertices of  $T^{K-\{w\}}$ are colored with $f(w)$.  
 Since $|T^{K-\{w\}}|$ is odd for every vertex $w \in K -\{v\}$,  and  $N(v) \cap I = \bigcup\limits_{w\in K\setminus \{v\}}  T^{K-\{w\}}$,
 each color from  $[k]\setminus \{f(v)\}$ 
 appears an even number of times in the neighborhood of $v$. This contradicts  the fact that $v$ has an odd color with respect to $f$. Therefore, $\chi_o(G)=k+1$. 
 
 \emph{Forward direction.}
 Suppose $\chi_o(G)=|K|+1$. We show that there exists a vertex $v\in K$ satisfying the conditions in the theorem. Assume, for contradiction, that no such vertex $v \in K$ exists.  
 We construct a coloring $f:V(G)\rightarrow [k]$ as follows. 
 First, we color the vertices of $K$ 
 with  $k$ distinct colors.  Then we color the vertices of $I$ such that each vertex in $I$ has an odd color with respect to $f$. 
Now we explain how to extend the partial coloring $f$ to the vertices of $I$ so that every vertex in $K$ has an odd color with respect to $f$.

Let $v$ be an arbitrary vertex of $K$. Then  at least one of the following condition must hold:
 \begin{enumerate}
  \item $|T^{K-\{w\}}|$ is even,  for some $w \in K-\{v\}$.
  
  \item $(N(v) \cap I) - \bigcup\limits_{w\in K\setminus \{v\}}  T^{K-\{w\}}\neq \emptyset$. 
 \end{enumerate}

 \begin{itemize}
     \item \textbf{Case~1.} $|T^{K-\{w\}}|$ is even for some $w \in K-\{v\}$

In this case  either $|T^{K-\{w\}}|=0$ or $|T^{K-\{w\}}|=2p$ for some $p\in \mathbb{N}$. If $|T^{K-\{w\}}|=0$.  
the color $f(w)$ serves as an odd color 
for every vertex of $K$ except for $w$.

If $|T^{K-\{w\}}|=2p$,  we color all vertices  of $T^{K-\{w\}}$ with $f(w)$. 
Since $|T^{K-\{w\}}|$ is even, the color $f(w)$ acts as an odd color (in the partial coloring that we have so far) for all vertices in $K$ except $w$. 

Next, we extend the coloring to the remaining vertices of $I$ so that $w$ get an odd color. 

We now have two cases based on $N(w) \cap I$. 
\begin{itemize}
    \item 
\textbf{Case~1a.} $|T^{K-\{z\}}|$ is even for some $z \in K-\{w\}$

If $|T^{K-\{z\}}|>0$, 
we color all vertices in $T^{K-\{z\}}$ with $f(z)$. Since $|T^{K-\{z\}}|$ is even, the color $f(z)$ 
acts as an odd color for $w$ (this holds for the case even when $|T^{K-\{z\}}|=0$). At this point, each vertex in $K$ has an odd color. We assign colors to the remaining vertices of $I$ in such a way that the odd colors for the  vertices in $K$ remain the same.

We assign each vertex of $|T^{K-\{y\}}|$, where $y\in K\setminus  \{w,z\}$, the color $f(y)$. 
 
The remaining uncolored vertices in $I$ have degree at most $k-2$. 
We color every uncolored vertex of $N(w)$ with any color other than color of $z$. Color rest of the uncolored vertices with 
any color other than $f(w)$.  
It is easy to verify that $f$ is an odd coloring. 

\item 
\textbf{Case~1b.} $|T^{K-\{z\}}|$ is odd for every $z \in K-\{w\}$ and  $(N(w) \cap I) - \bigcup\limits_{z\in K \setminus \{w\}}  T^{K-\{z\}}\neq \emptyset$. 

For every $w \in K$, we color the vertices in $T^{K\setminus \{w\}}$ is with the color $f(w)$. 
Let $D_w=(N(w) \cap I) - \bigcup\limits_{z\in K \setminus \{w\}}  T^{K-\{z\}}$. 
Observe that each vertex in $D$ has degree at most $k-2$. 
We choose an arbitrary vertex $w'\in D_w$ and assign  
a color  $c_w\in [k]-f(N(w'))$ and color the remaining vertices 
of $D_w$ with any available color other than $c_w$. 
Each uncolored vertex of $I\setminus D_w$ is assigned a color other than  $f(w)$. 
In this case $c_w$ is an odd color for $w$ and $f(w)$ is an odd color for every vertex $u \in K-\{w\}$. Thus $f$ is an odd coloring. 
\end{itemize}
\item 
\textbf{Case~2.} For every $v \in K$,  (i) $|T^{K-\{w\}}|$ is odd for every $w \in K-\{v\}$, and (ii) 
$(N(v) \cap I) - \bigcup_{w}  T^{K-\{w\}}\neq \emptyset$.
 
For every $w \in K$, each vertex in $T^{K\setminus \{w\}}$ is assigned the color $f(w)$. 
Thus, every color of $[k]\setminus \{f(w)\}$ appears an even number of times in the open neighborhood of each vertex $w\in K$, in the partial coloring $f$ constructed so far. We now assign colors to the remaining uncolored vertices of $I$ such that each vertex in $K$ has an odd color.

We extend the coloring $f$ using Algorithm~\ref{alg-1}. The input to the Algorithm is a split graph $G=(K,I)$ and a partial coloring $f$ as described above. The algorithm extends $f$ to an odd $k$-coloring of $G$.  \qed
 \end{itemize}
\end{proof}

\begin{algorithm}[]
\DontPrintSemicolon
\caption{An odd coloring of $G=(K,I)$ satisfying the assumptions of Case 2. }
\label{alg-1}

\KwIn{Split graph $G=(K,I)$, a partial coloring $f$ of $G$}
\KwOut{ An odd coloring $f$ of $G$ using $k$ colors}
{
$K'=K$ and $I'=I - \bigcup_{w\in K}  T^{K-\{w\}}$\;
$j=0$, $Q=\emptyset$,   $R=\emptyset$ \;
\While{$K' \neq \emptyset$}{
Let $u\in I'$ be the highest degree vertex in $G'=(K',I')$\;
$j=j+1$\;
$p_j=u$, $Q_j=N(u)$,   $R_j=\{v \in I'~|~ N_{G'}(v) \subseteq N_{G'}(u)\}$   \hspace{0.15cm}\; 
$K'=K'-Q_j$, $I'=I'-(\{p_j\} \cup R_j)$\; 
$\ell=\max\{j:Q_j \neq \emptyset\}$\; 
\tcc{The partial coloring $f$ is extended to uncolored independent set vertices as follows. }
$f(p_1)=c_1$, where $c_1 \in f(Q_\ell)$.

\For {{\bf each} $j=2$ to $\ell$}{
 $f(p_j)=c_j$, where $c_j \in f(Q_{j-1})\setminus f(N(p_j))$
 \; 
 }

\For {{\bf each} $j=1$ to $\ell$}{
For each $v \in R_j$, $f(v)=c'_j$,
where $c'_j=[k]\setminus (f(N(v))\cup \{f(p_1),\ldots, f(p_j)\})$\;
}
}

return (coloring $f$)

}

\end{algorithm}

\subsection{Subclasses of Bipartite Graphs}
As \oc{} is {\sf NP}-complete on bipartite graphs~\cite{ahn2022proper}, in this section we study its complexity on specific subclasses  of bipartite graphs. We show that \oc{} remains {\sf NP}-complete on both perfect elimination bipartite graphs and star convex bipartite graphs. We first provide the necessary definitions of the graph classes used in this section. 

A graph is \emph{bipartite} if its vertex set can be partitioned into two disjoint sets such that no two vertices in the same set are adjacent.  A bipartite graph $G = (X,Y,E)$ is said to be \emph{tree convex} if there exists  a tree $T = (X,F)$, such that, for every vertex $y \in Y$, the neighborhood of $y$ induces a subtree in $T$. If $T$ is a star, then $G$ is called \emph{star-convex
bipartite graph}.  An
edge $e = xy$ is called a bisimplicial edge if $N(x) \cup N(y)$ forms a complete bipartite subgraph. Given a bipartite graph $G=(X,Y,E)$, let $\sigma$ = $(x_1y_1, x_2y_2,\ldots, x_ky_k)$
be an ordering of pairwise non-adjacent edges of $G$. Let $S_j = \{x_1, x_2,\ldots, x_j\} \cup \{y_1, y_2,\ldots, y_j\}$ for $j=1,2, \ldots,k$ with $S_0 = \emptyset$. The ordering $\sigma$ is called a perfect edge elimination scheme for $G$ if the following conditions hold: (a) for each $j = 0, 1,\cdots, k-1$, the each edge $x_{j+1}y_{j+1}$ is bisimplicial in $G = [(X \cup Y) \backslash S_j]$. (b) There is no edge in $G = [(X \cup Y) \backslash S_k]$. A graph that admits a perfect edge elimination scheme is called a \emph{perfect elimination bipartite graph}. 

\newpage
\subsubsection{Perfect Elimination Bipartite Graphs.}
\begin{theorem}\label{thm:perfect-hard}
 \oc{} is {\sf NP}-complete on perfect elimination bipartite graphs. 
\end{theorem}
\begin{proof}
    We give a reduction from the \emph{proper $k$-coloring} problem, which is known to be {\sf NP}-complete for $k \geq 3$. 
Given an instance $(G,k)$ of the proper $k$-coloring problem, where  
$V(G)=\{v_1, v_2, \dots, v_n\}$,  
$|E(G)| = m$, and $k \geq 3$,
we  construct a perfect elimination bipartite graph $H$ from $G$ as follows.
\begin{itemize}
    \item For each $v_iv_j\in E(G)$, add a new vertex $e_{ij}$. Connect $e_{ij}$ to both $v_i$ and $v_j$ by adding the edges
    $\{v_ie_{ij}, v_je_{ij}\}$ and remove the original edge $v_iv_j$.

    \item For each $p \in [n]$, add a new vertex $y_p$ and connect it to $v_p$ by adding the edge $y_pv_p$. 
\end{itemize}
Formally,  $H = (V(G)\cup (V_E \cup V_Y), E_A \cup E_B)$, where $V_E = \{ e_{ij} \mid v_iv_j \in E(G)\}$, $V_Y = \{y_p \mid p \in [n]\}$, $E_A = \{v_ie_{ij}, v_je_{ij} \mid v_iv_j \in E(G)\}$ and $E_B = \{y_pv_p \mid p \in [n]\}$. Clearly, $H$ is a bipartite graph and the construction of $H$ can be done in polynomial time.

Note that $H$ is a perfect elimination bipartite graph,as the sequence  $(y_1v_1, y_2v_2,\ldots,y_nv_n)$ forms a perfect edge elimination scheme.

We now show that $(G, k)$ is a yes instance of proper $k$-coloring if and only if $(H, k)$ is a yes instance of the odd coloring problem using at most $k$ colors.

\emph{Forward direction.} 
Let $f:V(G) \rightarrow [k]$ be a proper coloring of $G$ using $k$ colors. We define a coloring $g:V(H) \rightarrow [k]$ of $H$ as follows:
For $v \in V(G)$, set $g(v) = f(v)$. For each $e_{ij} \in V_E$, set $g(e_{ij})=\ell_1$, where $\ell_1$ is an arbitrary color  chosen
from the set $[k] \setminus \{g(v_i),g(v_j)\}$. After coloring of $V(G) \cup V_E$, if a vertex $v_p \in V(G)$ has a color $c$ that appears an odd number of times in $N(v_p)\cap V_E$, set $g(y_p) \in [k] \setminus \{c\}$. If no such color $c$ exists, set $g(y_p) =\ell_2$, where $\ell_2$ is an arbitrary color used to color the vertices in $N(v_p) \cap V_E$, that is $\ell_2 \in g(N(v_p) \cap V_E)$.
It is easy to verify that $g$ is an odd coloring of $H$.

\emph{Reverse direction.}
Let $g$ be an odd coloring of $H$ using $\ell$ colors. We show that $g$ restricted to $V(G)$ gives a proper coloring of $G$ using at most $\ell$ colors. Since $N(e_{ij}) = \{v_i,v_j\}$, we have $g(v_i) \neq g(v_j)$ for each edge $v_iv_j \in E(G)$. Therefore, when $g$ restricted to the vertices of $G$ gives a proper coloring of $G$ using at most $\ell$ colors. 
\qed
\end{proof}

\subsubsection{Star-convex Bipartite Graphs}
\begin{theorem}\label{thm:star-hard}
\oc{} is {\sf NP}-complete on star-convex bipartite graphs.
\end{theorem}

\begin{proof}
    We give a polynomial time reduction from the \emph{proper $k$-coloring} problem, which is known to be NP-complete for $k\geq 3$. Given an instance $(G,k)$ of graph coloring, we define a new graph $G'$ with $V(G') = V(G) \cup \{x\}$ and $E(G') = E(G) \cup \{xv ~|~ v \in V(G)\}$. That is, $x$ is a universal vertex in the graph $G'$. Next, we construct a star-convex bipartite graph 
 $H$ from $G'$ as follows.

$$ V(H)= V(G')\cup\{w\} \cup \{e_{uv} ~|~ uv \in E(G')\} \text \qquad { and } $$ 
$$E(H)=  ~\{wy ~|~ y \in V(G')\} ~\cup ~\{ue_{uv},ve_{uv}\ ~|~ uv \in E(G') \}$$
    

    Formally, $H = (V(G') \cup V_E \cup \{w\}, E_A \cup E_B$), where $V_E = \{ e_{uv} \mid uv \in E(G')\}$, $E_A = \{ue_{uv}, ve_{uv} \mid uv \in E(G)\}$ and $E_B = \{wy \mid y \in V(G')\}$. Clearly, $H$ is a bipartite graph and can be constructed in polynomial time.

    We construct a star graph on $m+ 1$ vertices with $w$ as the center vertex and edges of $G'$ as leaves. Then the neighborhood of any vertex $v$ of $G'$ induces a subtree of the star graph as $w$ is adjacent to every vertex $v$. Hence, the graph $H$ is a star-convex bipartite graph.

    \begin{claim}
        $\chi(G) \leq k$ if and only if $\chi_o(H) \leq k + 2$.
    \end{claim}
    \begin{proof}
        \emph{Forward direction.} 
Let $f:V(G) \rightarrow [k]$ be a proper coloring of $G$ using $k$ colors. We define a coloring $g:V(H) \rightarrow [k+2]$ of $H$ as follows:
for $v \in V(G)$ set $g(v) = f(v)$.  For every $uv \in E(G)$, set $g(e_{uv})=\ell_1$, where $\ell_1$ is an arbitrarily chosen color   
 from the set $[k] \backslash \{g(u),g(v)\}$. Set $g(x) = k+1$ and $g(w) = k+2$. It is easy to verify that $g$ is an odd coloring of $H$. 

\emph{Reverse direction.}
Let $g$ be an odd coloring of $H$ using $\ell$ colors. We show that $g$ restricted to $V(G)$ gives a proper coloring of $G$ using at most $\ell -2$ colors. Since $N(e_{uv}) = \{u,v\}$ for each edge $uv \in E(G')$, it follows that, $g(u) \neq g(v)$ for every edge $uv \in E(G')$. Therefore, the restriction of $g$ to the vertices of $G'$ gives a proper coloring of $G'$ with $\ell$ colors. 
Without loss of generality, assume that $g(w) = \ell$. Since $w$ is adjacent
to every vertex $v \in V (G')$, it follows that $g(v) \neq \ell$. As $x$ is a universal vertex in $G'$, we have, $g(x) \neq g(v)$, for each $v \in V(G)$. Therefore, $g$ when restricted to the vertices of $G$ gives a
proper coloring of $G$ with at most $\ell - 2$ colors, hence $\chi(G) \leq \ell-2$. This completes the proof of the claim. 
\qed
\end{proof}
\qed
\end{proof}


    


\bibliography{BibFile}

\begin{thebibliography}{10}

\bibitem{abel2018conflict}
Zachary Abel, Victor Alvarez, Erik~D Demaine, S{\'a}ndor~P Fekete, Aman Gour,
  Adam Hesterberg, Phillip Keldenich, and Christian Scheffer.
\newblock Conflict-free coloring of graphs.
\newblock {\em SIAM Journal on Discrete Mathematics}, 32(4):2675--2702, 2018.

\bibitem{ahn2022proper}
Jungho Ahn, Seonghyuk Im, and Sang-il Oum.
\newblock The proper conflict-free $ k $-coloring problem and the odd $ k
  $-coloring problem are np-complete on bipartite graphs.
\newblock {\em arXiv preprint arXiv:2208.08330}, 2022.

\bibitem{ajwani2007conflict}
Deepak Ajwani, Khaled Elbassioni, Sathish Govindarajan, and Saurabh Ray.
\newblock Conflict-free coloring for rectangle ranges using o$(n^{0.382})$
  colors.
\newblock In {\em Proceedings of the nineteenth annual ACM symposium on
  Parallel algorithms and architectures}, pages 181--187. ACM, 2007.

\bibitem{bar2008deterministic}
Amotz Bar-Noy, Panagiotis Cheilaris, and Shakhar Smorodinsky.
\newblock Deterministic conflict-free coloring for intervals: from offline to
  online.
\newblock {\em ACM Transactions on Algorithms (TALG)}, 4(4):44, 2008.

\bibitem{bodlaender2011cross}
Hans~L Bodlaender, Bart~MP Jansen, and Stefan Kratsch.
\newblock Cross-composition: A new technique for kernelization lower bounds.
\newblock In {\em 28th International Symposium on Theoretical Aspects of
  Computer Science}, page 165, 2011.

\bibitem{bodlaender2011kernel}
Hans~L Bodlaender, St{\'e}phan Thomass{\'e}, and Anders Yeo.
\newblock Kernel bounds for disjoint cycles and disjoint paths.
\newblock {\em Theoretical Computer Science}, 412(35):4570--4578, 2011.

\bibitem{d2cfpt}
Anudhyan Boral, Marek Cygan, Tomasz Kociumaka, and Marcin Pilipczuk.
\newblock A fast branching algorithm for cluster vertex deletion.
\newblock {\em Theory Comput. Syst.}, 58(2):357--376, 2016.

\bibitem{caro2022remarks}
Yair Caro, Mirko Petru{\v{s}}evski, and Riste {\v{S}}krekovski.
\newblock Remarks on odd colorings of graphs.
\newblock {\em Discrete Applied Mathematics}, 321:392--401, 2022.

\bibitem{Cheilaris}
Panagiotis Cheilaris.
\newblock Conflict-free coloring ({Ph.D}. thesis).
\newblock {\em City University of New York}, 2009.

\bibitem{cho2023odd}
Eun-Kyung Cho, Ilkyoo Choi, Hyemin Kwon, and Boram Park.
\newblock Odd coloring of sparse graphs and planar graphs.
\newblock {\em Discrete Mathematics}, 346(5):113305, 2023.

\bibitem{corneil1985linear}
Derek~G. Corneil, Yehoshua Perl, and Lorna~K Stewart.
\newblock A linear recognition algorithm for cographs.
\newblock {\em SIAM Journal on Computing}, 14(4):926--934, 1985.

\bibitem{courcelle1992monadic}
Bruno Courcelle.
\newblock The monadic second-order logic of graphs {III}: Tree-decompositions,
  minors and complexity issues.
\newblock {\em RAIRO-Theoretical Informatics and Applications}, 26(3):257--286,
  1992.

\bibitem{courcellecw}
Bruno Courcelle and Stephan Olariu.
\newblock Upper bounds to the clique width of graphs.
\newblock {\em Discret. Appl. Math.}, 101(1-3):77--114, 2000.

\bibitem{cranston2023note}
Daniel~W Cranston, Michael Lafferty, and Zi-Xia Song.
\newblock A note on odd colorings of 1-planar graphs.
\newblock {\em Discrete Applied Mathematics}, 330:112--117, 2023.

\bibitem{cygan2015parameterized}
Marek Cygan, Fedor~V Fomin, {\L}ukasz Kowalik, Daniel Lokshtanov, D{\'a}niel
  Marx, Marcin Pilipczuk, Micha{\l} Pilipczuk, and Saket Saurabh.
\newblock {\em Parameterized algorithms}, volume~4.
\newblock Springer, 2015.

\bibitem{downey1999parameterized}
Rod~G Downey and Michael~Ralph Fellows.
\newblock {\em Parameterized complexity}, volume~3.
\newblock springer Heidelberg, 1999.

\bibitem{even2003conflict}
Guy Even, Zvi Lotker, Dana Ron, and Shakhar Smorodinsky.
\newblock Conflict-free colorings of simple geometric regions with applications
  to frequency assignment in cellular networks.
\newblock {\em SIAM Journal on Computing}, 33(1):94--136, 2003.

\bibitem{fabrici2023proper}
Igor Fabrici, Borut Lu{\v{z}}ar, Simona Rindo{\v{s}}ov{\'a}, and Roman
  Sot{\'a}k.
\newblock Proper conflict-free and unique-maximum colorings of planar graphs
  with respect to neighborhoods.
\newblock {\em Discrete Applied Mathematics}, 324:80--92, 2023.

\bibitem{fomin2010intractability}
Fedor~V Fomin, Petr~A Golovach, Daniel Lokshtanov, and Saket Saurabh.
\newblock Intractability of clique-width parameterizations.
\newblock {\em SIAM Journal on Computing}, 39(5):1941--1956, 2010.

\bibitem{robert-g-nd}
Robert Ganian.
\newblock Using neighborhood diversity to solve hard problems.
\newblock {\em CoRR}, abs/1201.3091, 2012.

\bibitem{gargano2015complexity}
Luisa Gargano and Adele~A Rescigno.
\newblock Complexity of conflict-free colorings of graphs.
\newblock {\em Theoretical Computer Science}, 566:39--49, 2015.

\bibitem{hickingbotham2022odd}
Robert Hickingbotham.
\newblock Odd colourings, conflict-free colourings and strong colouring
  numbers.
\newblock {\em arXiv preprint arXiv:2203.10402}, 2022.

\bibitem{lampis2012algorithmic}
Michael Lampis.
\newblock Algorithmic meta-theorems for restrictions of treewidth.
\newblock {\em Algorithmica}, 64(1):19--37, 2012.

\bibitem{lev2009conflict}
Nissan Lev-Tov and David Peleg.
\newblock Conflict-free coloring of unit disks.
\newblock {\em Discrete Applied Mathematics}, 157(7):1521--1532, 2009.

\bibitem{liu2022proper}
Chun-Hung Liu.
\newblock Proper conflict-free list-coloring, subdivisions, and layered
  treewidth.
\newblock {\em arXiv preprint arXiv:2203.12248}, 2022.

\bibitem{nagy2008online}
Judit Nagy-Gy{\"o}rgy and Cs~Imreh.
\newblock Online hypergraph coloring.
\newblock {\em Information Processing Letters}, 109(1):23--26, 2008.

\bibitem{pach2009conflict}
J{\'a}nos Pach and G{\'a}bor Tardos.
\newblock Conflict-free colourings of graphs and hypergraphs.
\newblock {\em Combinatorics, Probability and Computing}, 18(5):819--834, 2009.

\bibitem{petr2023odd}
Jan Petr and Julien Portier.
\newblock The odd chromatic number of a planar graph is at most 8.
\newblock {\em Graphs and Combinatorics}, 39(2):28, 2023.

\bibitem{petruvsevski2022colorings}
Mirko Petru{\v{s}}evski and Riste {\v{S}}krekovski.
\newblock Colorings with neighborhood parity condition.
\newblock {\em Discrete Applied Mathematics}, 321:385--391, 2022.

\end{thebibliography}


\end{document}